\documentclass[11pt,a4paper]{article}

\usepackage[margin=1in]{geometry}
\usepackage{amsmath,amssymb,amsthm,mathtools}
\usepackage{bm}
\usepackage{booktabs}
\usepackage{graphicx}
\usepackage{microtype}
\usepackage{hyperref}
\usepackage{cleveref}
\usepackage[numbers,sort&compress]{natbib}

\newtheorem{definition}{Definition}[section]

\newtheorem{proposition}[definition]{Proposition}
\newtheorem{theorem}[definition]{Theorem}

\newtheorem{conjecture}[definition]{Conjecture}

\theoremstyle{remark}
\newtheorem{remark}[definition]{Remark}

\newtheorem{prediction}[definition]{Experimental prediction}

\newcommand{\R}{\mathbb{R}}

\newcommand{\E}{\mathbb{E}}
\newcommand{\Prob}{\mathbb{P}}
\newcommand{\Var}{\operatorname{Var}}
\newcommand{\Cov}{\operatorname{Cov}}
\newcommand{\MI}{\operatorname{I}}
\newcommand{\TE}{\operatorname{TE}}

\newcommand{\Tr}{\operatorname{Tr}}
\newcommand{\diag}{\operatorname{diag}}

\newcommand{\dd}{\,\mathrm{d}}

\newcommand{\state}{\bm{x}}
\newcommand{\stimulus}{\bm{u}}
\newcommand{\response}{\bm{y}}
\newcommand{\weights}{\bm{G}}

\newcommand{\laplacian}{\bm{L}}
\newcommand{\noise}{\bm{\xi}}

\title{\textbf{Mathematical formalisation of information processing and intelligence in proteinoid networks}}

\author{
Andrew Adamatzky and Panagiotis Mougkogiannis\\
Unconventional Computing Laboratory\\
University of the West of England, Bristol, UK
}

\date{}

\begin{document}

\maketitle

% ------------------------------------------------------------
% Abstract
% ------------------------------------------------------------

\begin{abstract}

Proteinoid microspheres are electrically active assemblies of thermally polymerised amino acids that exhibit spontaneous oscillations, spike-like events, stimulus-dependent responses, adaptation and collective dynamics.  We formulate these observations as an applied-mathematical problem and develop a unified dynamical framework for information processing in  proteinoid networks. Individual microspheres are represented by nonlinear excitable or oscillatory systems, whereas interactions are described by weighted adaptive graphs and, in the spatial limit, reaction--diffusion dynamics. The framework combines signal propagation, attenuation, fading memory, adaptive coupling, synchronisation, information transfer, reservoir
readout, robustness and energetic constraints. We define a task-relative proteinoid intelligence functional from normalised measures of stimulus information, memory, adaptation, directed information transfer, task performance and robustness, penalised by energetic and structural costs. Analytical results establish a joint-response information bound, an upper bound on linear memory for stable linear networks, and existence of an
optimal admissible coupling matrix under compact constraints. Dimensionless groups identify the relative time scales and coupling regimes governing the problem in nonlinear and adaptive network dynamics.

\end{abstract}


\section{Introduction}

Intelligence is conventionally associated with nervous systems, brains and artificial computational architectures. This association leaves open a more fundamental question: what is the minimal physical organisation required for a material system to sense, remember, adapt and generate context-dependent responses? The question is increasingly relevant to work on basal cognition, embodied cognition, morphological computation, active matter and material intelligence, all of which weaken the assumption that sophisticated information processing must begin with a nervous system \citep{PfeiferBongard2007,Lyon2015,Levin2019,Stepney2018,Nakajima2020,FieldsLevin2020,Dennett2017}. In these approaches, cognition-like capacities are treated not as all-or-nothing attributes but as experimentally separable abilities to detect regularities, retain state, coordinate distributed processes and alter future behaviour.

A complementary perspective comes from the physics of computation. Computation is necessarily instantiated in matter and therefore constrained by dynamics, noise, dissipation, memory and energy \citep{Landauer1961,Bennett1982}. This observation motivates physical and unconventional computing, in which the natural dynamics of a substrate perform transformations that would otherwise have to be explicitly simulated by conventional electronics \citep{Adamatzky2017,Stepney2018,Tanaka2019,Nakajima2020}. Chemical reaction--diffusion media, soft matter, oscillator networks and other nonlinear substrates demonstrate that computation can be distributed through the state evolution of matter rather than confined to a sequence of symbolic operations \citep{Murray2002,KeenerSneyd2009}. Proteinoids are especially interesting in this context because they combine chemical simplicity, spontaneous self-assembly and experimentally observable electrical dynamics.

\begin{figure}[!tbp]
    \centering
    \includegraphics[width=0.49\linewidth]{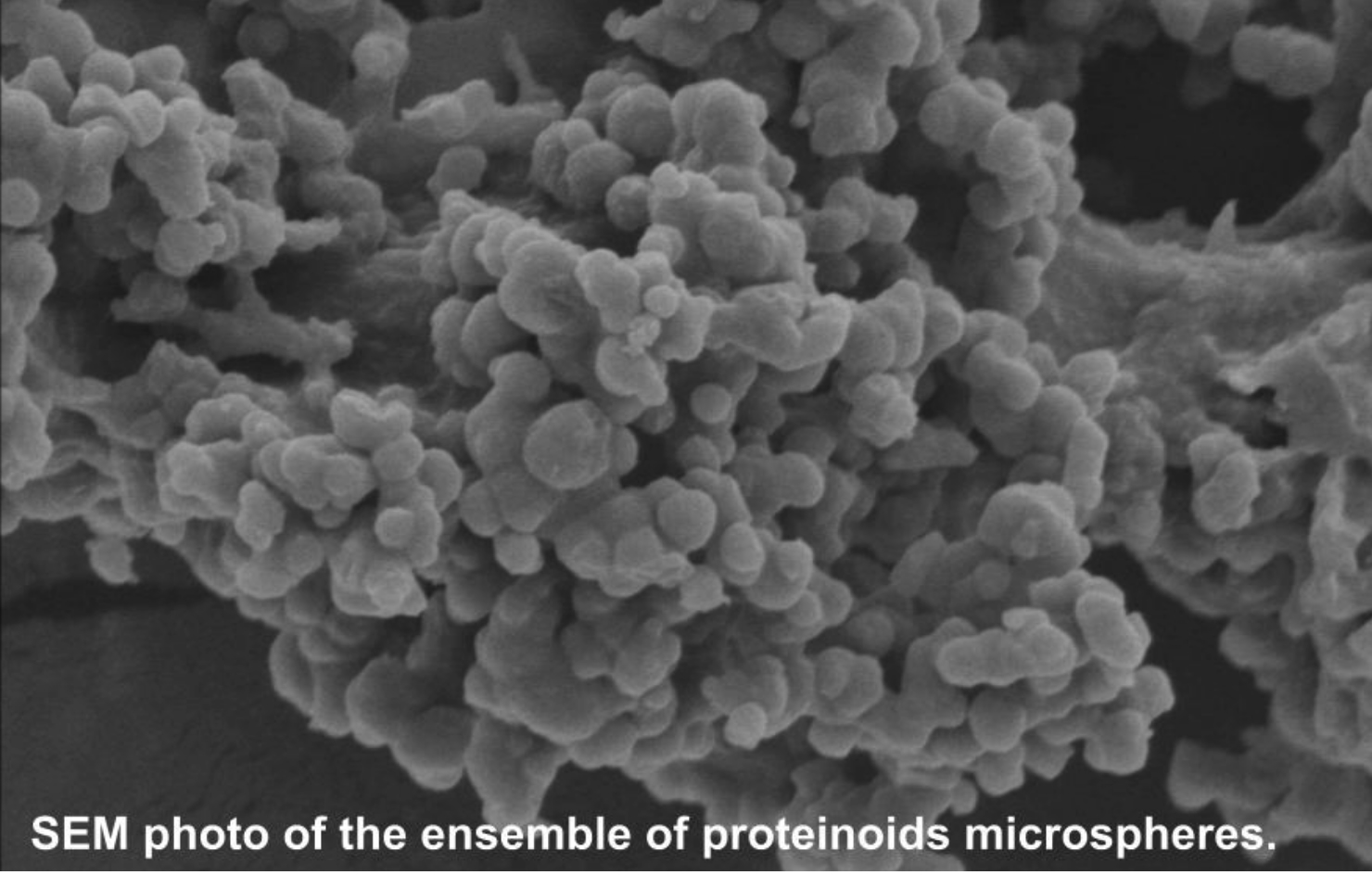}
     \includegraphics[width=0.49\linewidth]{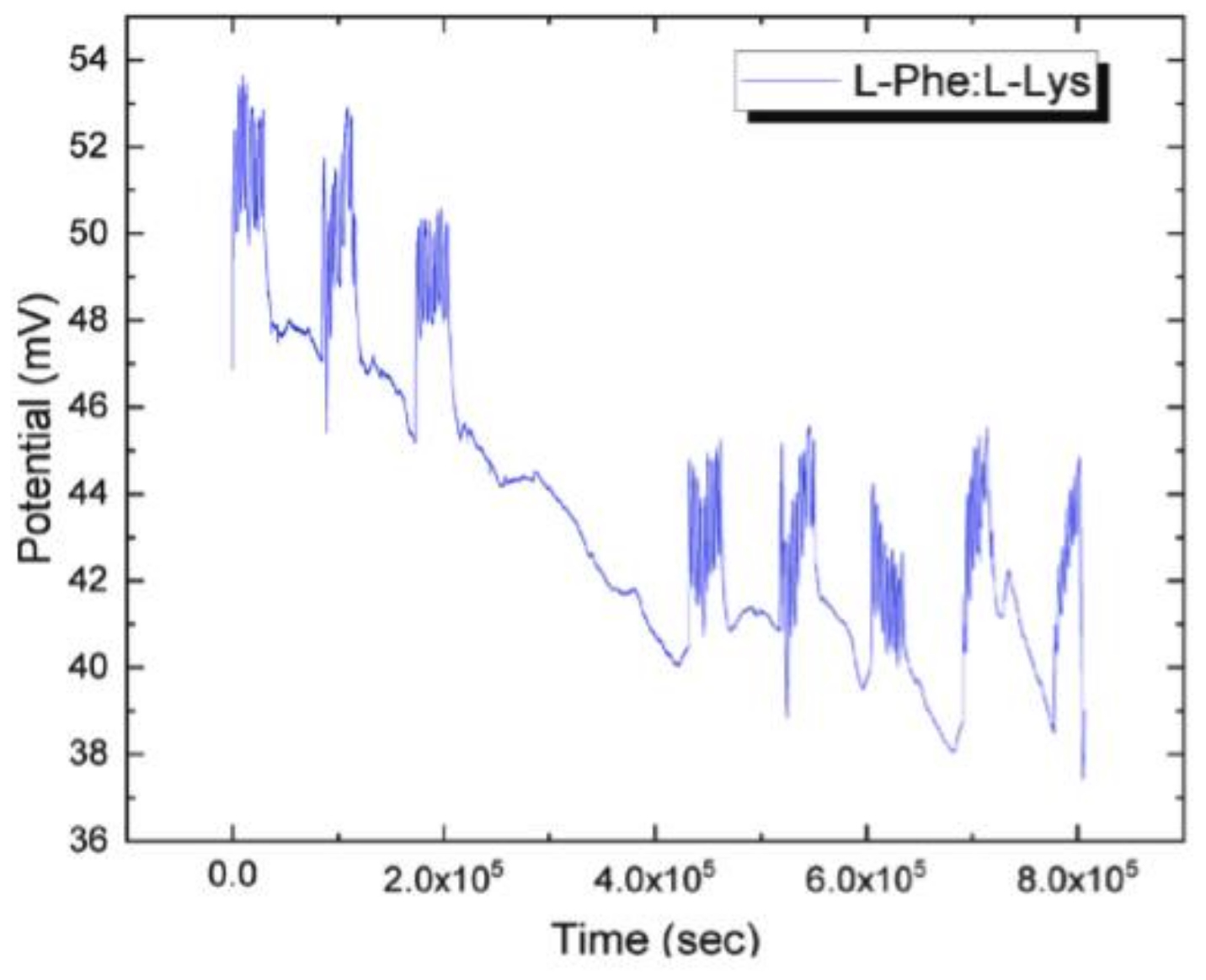}
    \caption{(a) Scanning electron microscopy (SEM) image of an ensemble of proteinoid microspheres. (b) Neuron-like chattering spiking activity recorded from an ensemble of proteinoid microspheres.}
    \label{fig:proteinoids}
\end{figure}

Proteinoids are thermally synthesised amino-acid polymers that can self-assemble into microspherical structures \citep{FoxHarada1958,FoxHarada1960,FoxHaradaVegotsky1959,Fox1965,FoxDose1977,Fox1992} (Fig.~\ref{fig:proteinoids}a). Their historical importance lies in research on chemical evolution and protocellular organisation, where thermal amino-acid polymers and microspheres were investigated as models for how organised, compartmentalised and dynamically active structures might arise without contemporary biological machinery \citep{Fox1965,FoxDose1977,Fox1992}. The present work shifts the emphasis from origins-of-life morphology to dynamics and computation: if a chemically simple microspherical assembly exhibits excitable and history-dependent electrical behaviour, it becomes possible to ask what forms of information processing can be realised before the appearance of neurons.

Experimental studies indicate that proteinoid ensembles can exhibit spontaneous electrical activity, oscillations, spike-like events, nonlinear responses to stimulation, adaptation and collective behaviour \citep{Przybylski1985,MougkogiannisAdamatzky2023waves,MougkogiannisAdamatzky2023logic,MougkogiannisAdamatzky2023electrical,MougkogiannisAdamatzky2024proto} (Fig.~\ref{fig:proteinoids}b). Logical operations have also been demonstrated in ensembles of proteinoid microspheres \citep{MougkogiannisAdamatzky2023logic}. These observations motivate the possibility that proteinoid systems provide experimentally accessible models of information processing emerging in non-neural matter. Importantly, the claim does not depend on treating every electrical spike as a cognitive event. The relevant issue is whether the dynamics support reproducible mappings between environmental perturbations, internal state trajectories and subsequent responses.

The distinction between activity and information is fundamental. A material may oscillate strongly while conveying little information about its environment; conversely, relatively small electrical changes may be highly informative if they are stimulus-specific and reproducible. Information theory therefore provides a natural language for separating mere activity from stimulus encoding, storage and transfer \citep{Shannon1948,CoverThomas2006,Schreiber2000,Lizier2012}. Likewise, dynamical-systems theory provides tools for distinguishing transient excitation, stable oscillation, synchronisation, multistability and history dependence \citep{Strogatz2000,Izhikevich2007}. These distinctions are essential if the word intelligence is to acquire quantitative rather than metaphorical meaning in a chemical substrate.

Collective oscillation is particularly relevant because proteinoid microspheres need not act as independent units. The classical theory of coupled oscillators developed by Kuramoto provides a minimal mathematical language for the emergence of macroscopic coordination from interacting oscillatory elements. Kuramoto's early analysis of self-entrainment in populations of coupled nonlinear oscillators established the canonical phase-coupling framework \citep{Kuramoto1975}; his subsequent monograph developed the statistical and dynamical theory of collective synchronisation \citep{Kuramoto1984}. Later work by Kuramoto and Battogtokh showed that nonlocal coupling can support coexistence of coherent and incoherent domains \citep{KuramotoBattogtokh2002}, demonstrating that collective organisation need not collapse into uniform global synchrony. These ideas, together with later analyses of the Kuramoto model and synchronisation on complex networks \citep{Acebron2005,Arenas2008,Rodrigues2016}, are directly relevant to proteinoid ensembles because useful computation may require a balance between coordinated propagation and preservation of differentiated local dynamics.

A second theoretical bridge is provided by reservoir computing. In echo-state and liquid-state approaches, a nonlinear dynamical medium transforms temporal inputs into high-dimensional transient states from which comparatively simple readouts can recover task-relevant information \citep{Jaeger2001,Maass2002,LukoseviciusJaeger2009}. Physical reservoir computing generalises this principle from artificial recurrent networks to material substrates whose intrinsic dynamics supply nonlinearity, fading memory and state-space expansion \citep{Tanaka2019,Nakajima2020}. This perspective is attractive for proteinoids because it does not require the microspheres to implement neuron-like components. It requires only that stimulation produces reproducible, history-dependent trajectories with sufficient separation and fading memory for a readout to exploit them.

The objective of this paper is not to attribute consciousness, cognition or semantic understanding to proteinoids. We use intelligence in a deliberately operational sense, asking whether a material system can discriminate between stimuli, transform and propagate signals, retain traces of previous states, modify later responses as a consequence of experience, integrate information across a network, and sustain task-dependent behaviour in the presence of noise and physical constraints. This usage is closer to a graded notion of basal or minimal cognition than to human intelligence \citep{Lyon2015,Levin2019,FieldsLevin2020}. These capacities are treated as measurable properties of dynamics rather than as analogies with nervous systems.

The central hypothesis is that intelligence in proteinoid networks does not reside in any single microsphere or single electrical event. This distributed interpretation is compatible with embodied and material-computation perspectives \citep{PfeiferBongard2007,Stepney2018,Nakajima2020}. It emerges from the interaction of nonlinear excitability, adaptive coupling, memory and network-level information integration. The mathematical development therefore moves from local state dynamics to propagation, history dependence, collective organisation and information processing, and finally to a task-relative functional intended to make comparisons between different proteinoid systems possible. In this sense, the paper asks not whether proteinoids are intelligent in an anthropomorphic sense, but under what measurable dynamical and informational conditions it becomes scientifically useful to describe an excitable chemical network as exhibiting a primitive form of intelligence.

From an applied-mathematics perspective, the contribution is fourfold. First,
we place local nonlinear dynamics, network coupling and spatial propagation
within a common state-space formulation. Second, we introduce explicit
variables for fading memory and adaptive coupling so that history dependence
is represented dynamically rather than inferred from response curves alone.
Third, we connect collective dynamics to measurable information-processing
quantities and to a constrained task-relative objective functional. Fourth,
we derive analytical bounds and dimensionless control parameters and specify
a parameter-estimation and numerical-validation programme. The emphasis is
therefore on a model class whose assumptions, observables and predictions can
be confronted with experimental time series.

\section{Proteinoid networks as dynamical systems}

Let a proteinoid network consist of \(N\) interacting microspheres. The internal state of the network at time \(t\) is \(\state(t) = \bigl( \bm{x}_1(t), \bm{x}_2(t), \ldots, \bm{x}_N(t) \bigr),\)
where \(\bm{x}_i(t)\in\R^m\) denotes the physiological and electrical state
of microsphere \(i\). 
The external stimulus is represented by
\(\stimulus(t)\in\mathcal{U},\)
and the measurable response by
\(\response(t) = h\bigl(\state(t)\bigr) + \bm{\varepsilon}(t),\)
where \(h\) is an observation function and \(\bm{\varepsilon}(t)\) represents
measurement noise.

\begin{definition}[Operational proteinoid intelligence]
A proteinoid network is operationally intelligent with respect to a class of
stimuli \(\mathcal{U}\) and a class of tasks \(\mathcal{T}\) if it can:  acquire statistically significant information about external stimuli; retain information about previous stimuli or internal states; modify its future responses as a consequence of experience; produce distinguishable and task-relevant outputs; maintain these capacities under perturbation and noise.
\end{definition}

This definition is functional rather than anthropomorphic. It does not require representation, consciousness or symbolic reasoning. The operational definition becomes useful only when it is attached to a physical state space. We therefore represent the same system as a dynamical graph, so that the electrical state of each microsphere and the changing interactions between microspheres enter a common mathematical description.

Following standard complex-network and adaptive-network formalisms \citep{Newman2003,Boccaletti2006,GrossBlasius2008}, a proteinoid network is represented by a weighted graph
\(\mathcal{G}(t) = \bigl( \mathcal{V}, \mathcal{E}(t), \weights(t) \bigr),\)
where
\(\mathcal{V} = \{1,2,\ldots,N\}\)
is the set of microspheres and
\(\weights(t) = \bigl[g_{ij}(t)\bigr]_{i,j=1}^{N}\)
is the time-dependent coupling matrix.
The graph Laplacian is
\(\laplacian(t) = \diag\left( \sum_{j=1}^{N} g_{ij}(t) \right) - \weights(t).\)
The state of microsphere \(i\) evolves according to
\(\frac{\dd \bm{x}_i}{\dd t} = \bm{F}_i \bigl( \bm{x}_i, \bm{u}_i, \bm{\theta}_i \bigr) + \sum_{j=1}^{N} g_{ij}(t) \bm{H} \bigl( \bm{x}_j-\bm{x}_i \bigr) + \bm{\xi}_i(t),\)
Here \(\bm{F}_i\) describes the intrinsic dynamics of microsphere \(i\), \(\bm{u}_i(t)\) is its local stimulus, \(\bm{\theta}_i\) collects the parameters governing its response, \(\bm{H}\) specifies how the state of a
neighbour enters the coupling term, and \(\bm{\xi}_i(t)\) represents stochastic fluctuations. This decomposition is intentionally broad: it allows the same network description to accommodate excitable, oscillatory or mixed
local dynamics while keeping the topology and strength of interactions explicit. In matrix form,
\(\dot{\state} = \bm{F}(\state,\stimulus,\bm{\theta}) - \bigl( \laplacian(t)\otimes \bm{K} \bigr) \state + \noise(t).\)

The graph formalism specifies how microspheres are connected but not how an individual microsphere responds. At the local level we require dynamics capable of reproducing thresholded excitation, recovery, spontaneous oscillation and nonlinear stimulus dependence. The following models are intended as minimal phenomenological descriptions rather than claims about a unique microscopic mechanism.
A minimal two-variable excitable model, motivated by the FitzHugh--Nagumo class \citep{FitzHugh1961,Nagumo1962,Izhikevich2007}, is
\(\tau_V \frac{\dd V_i}{\dd t} = V_i - \frac{V_i^3}{3} - W_i + I_i(t) + \sum_{j=1}^{N} g_{ij}(V_j-V_i),;\quad \tau_W \frac{\dd W_i}{\dd t} = V_i + a_i - b_i W_i,\)
where \(V_i\) is an electrical variable and \(W_i\) is a slow recovery or adaptation variable.

The input current is decomposed as
\(I_i(t) = I_i^{\mathrm{base}} + I_i^{\mathrm{stim}}(t) + \sigma_i \xi_i(t).\) When the dominant observation is oscillatory rather than pulse-like, the same local dynamics can be represented at a coarser level by phase and amplitude variables. For oscillatory microspheres, a phase-amplitude model may be used, drawing on standard coupled-oscillator descriptions \citep{Kuramoto1975,Kuramoto1984,Acebron2005,Arenas2008,Rodrigues2016}:

\(\dot{r}_i = \mu_i r_i - r_i^3 + \sum_{j=1}^{N} g_{ij}r_j \cos(\phi_j-\phi_i),;\quad \dot{\phi}_i = \omega_i + \sum_{j=1}^{N} k_{ij} \sin(\phi_j-\phi_i) + \Gamma_i u_i(t).\)

Independent of the particular state-space representation, the experimentally accessible input--output relation can be summarised by a saturating response law.
A saturating stimulus--response relationship can be represented by a Hill-type law \citep{Hill1910}:

\begin{equation}
R(u) = R_0 + \frac{R_{\max}u^n}{K^n+u^n}, \label{eq:hill-response}
\end{equation}

where \(K\) is the half-response stimulus and \(n\) controls the degree of
nonlinearity.

\begin{definition}[Excitability]
Microsphere \(i\) is excitable if a subthreshold perturbation decays towards
a stable state, whereas a suprathreshold perturbation produces a large
transient excursion before recovery.
\end{definition}

\section{Signal propagation}

For two recording locations \(i\) and \(j\), the propagation delay is approximated by
\(\tau_{ij} \approx \frac{d_{ij}^{\mathrm{eff}}}{v_{ij}},\)
where \(d_{ij}^{\mathrm{eff}}\) is the effective path length and \(v_{ij}\) is the propagation velocity.
Propagation is inseparable from loss. A signal that reaches a distant part of the network with a measurable delay will generally also arrive with reduced amplitude.

The signal amplitude at distance \(d\) may be described by
\begin{equation}
A(d) = A_0 \exp\left(-\frac{d}{\lambda}\right), \label{eq:attenuation}
\end{equation}
where \(\lambda\) is the characteristic attenuation length.
A more general power-exponential law is
\(A(d) = A_0 d^{-\alpha} \exp\left(-\frac{d}{\lambda}\right).\)

The preceding expressions treat propagation phenomenologically. When the proteinoid assembly is regarded as a spatially distributed excitable medium, a reaction--diffusion representation provides a continuous counterpart \citep{Murray2002,KeenerSneyd2009}.
For a spatially distributed proteinoid medium,

\(\frac{\partial V}{\partial t} = D_V\nabla^2V + f(V,W) + I(\bm{r},t),;\quad \frac{\partial W}{\partial t} = D_W\nabla^2W + g(V,W).\)

\begin{conjecture}[Propagation threshold]
Suppose the local dynamics are excitable and the coupling between adjacent microspheres exceeds a critical value \(g_c\). Then a suprathreshold event initiated at one node can generate a propagating excitation through a
connected component of the network. The value of \(g_c\) is model-dependent and should be determined by travelling-wave analysis, numerical continuation or direct experimental estimation rather than treated as a universal constant.
\end{conjecture}

\section{Memory and adaptation}

Propagation provides spatial integration, but intelligence in the operational
sense adopted here also requires temporal dependence. We therefore introduce
adaptation, recovery and memory directly into the dynamical description.

The response to a persistent stimulus may follow

\(R(t\mid u) = R_{\infty}(u) + \bigl( R_0(u)-R_{\infty}(u) \bigr) \exp\left(-\frac{t}{\tau_a}\right),\)

where \(\tau_a\) is the adaptation time.
After the stimulus is removed, the system need not return instantaneously to its original response regime. Recovery can be characterised by the response to a later test stimulus.
The response to a test stimulus applied after a delay \(\Delta t\) is

\begin{equation}
R_{\mathrm{test}}(\Delta t) = R_{\mathrm{base}} - \Delta R \exp\left(-\frac{\Delta t}{\tau_r}\right), \label{eq:recovery-law}
\end{equation}

where \(\tau_r\) is the recovery time.
A convenient way to represent history dependence is to introduce an internal memory variable that integrates previous stimulation and feeds back onto excitability. Such fading-memory descriptions are also central to nonlinear systems theory and reservoir computing \citep{BoydChua1985,Jaeger2001,Maass2002}.
A memory state \(m_i(t)\) may evolve according to

\begin{equation}
\tau_m\dot{m}_i = -m_i + \Phi\bigl(u_i(t),V_i(t)\bigr). \label{eq:memory-state}
\end{equation}

The internal memory modifies excitability:
\(a_i(t) = a_i^{(0)} + \eta_m m_i(t).\)
Equivalently, memory may be expressed as a convolution:
\(m_i(t) = \int_{-\infty}^{t} K_m(t-s)u_i(s)\dd s,\)
with memory kernel
\(K_m(t) = \frac{1}{\tau_m} \exp\left(-\frac{t}{\tau_m}\right) \mathbf{1}_{t\geq 0}.\)

\begin{definition}[Fading memory]
A proteinoid system has fading memory if the influence of a stimulus applied
at time \(t-s\) on the present state decreases as \(s\rightarrow\infty\).
\end{definition}

\begin{proposition}[Exponential memory decay]
For the linear memory equation in \cref{eq:memory-state}, the contribution of
an isolated stimulus to \(m_i(t)\) decays exponentially with characteristic
time \(\tau_m\).
\end{proposition}

\section{Adaptive coupling and network plasticity}

Memory need not be confined to the internal state of a microsphere. Repeated
activity may also modify the effective couplings of the network, producing a
slower structural memory superimposed on the faster electrical dynamics,
analogous at a formal level to activity-dependent and adaptive-network
plasticity \citep{Hebb1949,GrossBlasius2008}.
The coupling between microspheres evolves according to

\begin{equation}
\frac{\dd g_{ij}}{\dd t} = \eta \left|q_{ij}(t)\right|^{\alpha} - \lambda_g g_{ij}, \label{eq:reinforcement}
\end{equation}

where \(q_{ij}(t)\) is the effective signal flow from node \(j\) to node
\(i\).

A correlation-based rule is
\(\frac{\dd g_{ij}}{\dd t} = \eta \Psi\bigl(V_i,V_j\bigr) - \lambda_g g_{ij},\)
with, for example, \(\Psi(V_i,V_j) = V_iV_j.\)
Unbounded reinforcement is physically implausible, so the plasticity law can be constrained by a finite coupling scale.

To ensure \(0\leq g_{ij}\leq g_{\max}\), we have

\(\frac{\dd g_{ij}}{\dd t} = \eta \left(1-\frac{g_{ij}}{g_{\max}}\right) \left|q_{ij}\right|^{\alpha} - \lambda_g g_{ij}.\)

At still longer timescales, plasticity may alter not only the strength but also the existence of connections.
The probability of forming a new connection is

\(\Prob\bigl( A_{ij}(t+\Delta t)=1 \mid A_{ij}(t)=0 \bigr) = 1- \exp\left[ -\rho_{ij}(t)\Delta t \right],\)

where

\(\rho_{ij}(t) = \rho_0 \exp\left(-\frac{d_{ij}}{\ell_g}\right) \Theta\bigl(C_{ij}(t)-C_c\bigr).\)

Connections may be pruned when
\(g_{ij}(t)<g_{\mathrm{prune}}.\)

\begin{proposition}[Activity-dependent path stabilisation]
Under \cref{eq:reinforcement}, pathways that repeatedly carry larger signal flows acquire greater asymptotic coupling strength than inactive pathways, provided reinforcement exceeds decay.
\end{proposition}

\section{Collective dynamics and synchronisation}

For phase variables,
\(\dot{\phi}_i = \omega_i + \sum_{j=1}^{N} g_{ij} \sin(\phi_j-\phi_i).\)
The Kuramoto order parameter, originating in Kuramoto's theory of collective phase synchronisation \citep{Kuramoto1975,Kuramoto1984,Acebron2005,Arenas2008,Rodrigues2016}, is
\(re^{\mathrm{i}\psi} = \frac{1}{N} \sum_{j=1}^{N} e^{\mathrm{i}\phi_j},\)
where \(r\in[0,1]\) measures global synchronisation.
Synchronisation is only one aspect of collective organisation. A useful network should also balance integration between different regions with preservation of differentiated local dynamics.
Let the network be partitioned into modules
\(\mathcal{M}_1,\ldots,\mathcal{M}_K\). A simple integration measure is

\(\mathcal{F}_{\mathrm{int}} = \frac{ \sum_{a\neq b} \sum_{i\in\mathcal{M}_a} \sum_{j\in\mathcal{M}_b} g_{ij} }{ \sum_{i,j}g_{ij} }.\)

A segregation measure may be based on network modularity \citep{NewmanGirvan2004,Newman2006}:

\(Q = \frac{1}{2m} \sum_{i,j} \left( g_{ij} - \frac{k_ik_j}{2m} \right) \delta(c_i,c_j).\)

The balance between propagation and diversity suggests examining the network near transitions between weak and globally coordinated activity, motivated by work linking critical regimes to dynamic range and information processing \citep{BeggsPlenz2003,KinouchiCopelli2006,ShewPlenz2013}.
Let the branching ratio be

\(\sigma = \frac{ \E[\text{number of events at }t+\Delta t] }{ \E[\text{number of events at }t] }.\)

The regimes are

\(\sigma<1 \quad\text{subcritical}, \qquad \sigma=1 \quad\text{critical}, \qquad \sigma>1 \quad\text{supercritical}.\)

\begin{conjecture}[Critical intelligence conjecture]
Proteinoid networks exhibit maximal dynamic range and information transmission near a transition between subcritical and supercritical collective activity.
\end{conjecture}

\section{Information processing}

Collective dynamics become computationally relevant when they carry
information about stimulation, retain information about past states and
transmit information between parts of the network. These properties can be
quantified without assigning semantic content to the electrical signals.
Let \(U\) denote the stimulus class and \(Y\) the measured response. Using the standard information-theoretic definition \citep{Shannon1948,CoverThomas2006}, the mutual information is

\begin{equation}
\MI(U;Y) = \sum_{u,y} p(u,y) \log \frac{p(u,y)}{p(u)p(y)}. \label{eq:mutual-information}
\end{equation}

For continuous variables,
\(\MI(U;Y) = \int p(u,y) \log \frac{p(u,y)}{p(u)p(y)} \dd u\,\dd y.\)

Mutual information measures dependence between stimulus and response but does not identify directed interactions inside the network. For this purpose we use transfer entropy \citep{Schreiber2000,Bossomaier2016}.
The transfer entropy from node \(j\) to node \(i\) is

\(\TE_{j\rightarrow i} = \sum p \left( x_i^{t+1}, \bm{x}_i^{(k)}, \bm{x}_j^{(l)} \right) \log \frac{ p \left( x_i^{t+1} \mid \bm{x}_i^{(k)}, \bm{x}_j^{(l)} \right) }{ p \left( x_i^{t+1} \mid \bm{x}_i^{(k)} \right) }.\)

The total directed information flow is 
\(\mathcal{T} = \sum_{i\neq j} \TE_{j\rightarrow i}.\)
Directed transfer should be distinguished from storage. A node may carry predictive information about its own future even when little information is being transferred from its neighbours.
Following the active-information-storage formalism \citep{Lizier2012}, active information storage at node \(i\) is
\(\mathcal{A}_i = \MI \left( \bm{x}_i^{(k)}; x_i^{t+1} \right).\)
Network information storage is
\(\mathcal{A} = \frac{1}{N} \sum_{i=1}^{N}\mathcal{A}_i.\)

For experimental classification tasks it is also useful to quantify directly how well response distributions produced by different stimuli separate in feature space.
For stimulus classes \(u_a\) and \(u_b\), define the separability
\(D_{ab} = \frac{ \left\| \bm{\mu}_a-\bm{\mu}_b \right\|^2 }{ \Tr(\bm{\Sigma}_a+\bm{\Sigma}_b) },\)
where \(\bm{\mu}_a\) and \(\bm{\Sigma}_a\) are the mean response and
covariance under stimulus \(u_a\).

\begin{definition}[Discriminative capacity]
The discriminative capacity of the network over a stimulus set
\(\mathcal{U}\) is

\(\mathcal{D} = \frac{2}{ |\mathcal{U}| \bigl( |\mathcal{U}|-1 \bigr) } \sum_{a<b}D_{ab}.\)
\end{definition}

The fading-memory variables introduced above describe a mechanism; memory capacity describes what that mechanism permits the network to reconstruct. Let \(u(t-k)\) be a delayed input and \(\widehat{u}_k(t)\) its reconstruction from the current network state.
The memory coefficient at delay \(k\) is

\(MC_k = \frac{ \Cov^2 \left( u(t-k), \widehat{u}_k(t) \right) }{ \Var\bigl(u(t-k)\bigr) \Var\bigl(\widehat{u}_k(t)\bigr) }.\)

Following the memory-capacity approach developed for recurrent and reservoir systems \citep{Jaeger2002memory,White2004,Dambre2012}, the total linear memory capacity is
\(MC = \sum_{k=1}^{K}MC_k.\)
A nonlinear memory capacity may be defined by reconstructing nonlinear
functions of previous inputs:
\(NMC = \sum_{k=1}^{K} \sum_{p=2}^{P} MC_{k,p}.\)

\begin{conjecture}[Network memory advantage]
Under non-degenerate coupling and heterogeneous relaxation times, the total
memory capacity of a proteinoid network may exceed the mean memory capacity
of its isolated microspheres.
\end{conjecture}

\section{Learning and experience-dependent dynamics}

Let \(R_n(u)\) denote the response to the \(n\)-th presentation of stimulus
\(u\). Adaptation is quantified by
\(\mathcal{A}_{\mathrm{resp}} = \frac{ \left| R_N(u)-R_1(u) \right| }{ |R_1(u)|+\varepsilon }.\)
A particularly simple form of experience-dependent change is habituation, in which repeated presentation of the same stimulus reduces the response.
A simple habituation law is
\(R_n = R_{\infty} + \bigl( R_1-R_{\infty} \bigr) e^{-(n-1)/n_a}.\)
The opposite tendency, sensitisation, corresponds to progressive enhancement of the response under repeated stimulation. Sensitisation may be represented by
\(R_n = R_1 + \Delta R_{\max} \left( 1-e^{-(n-1)/n_s} \right).\)

History dependence becomes more interesting when the response to one stimulus is modified by its repeated temporal association with another.
For two stimuli \(u_a\) and \(u_b\), associative modification is present if
\(p \left( Y\mid u_a,\mathcal{H}_{ab} \right) \neq p \left( Y\mid u_a,\mathcal{H}_{a} \right),\)
where \(\mathcal{H}_{ab}\) denotes a history in which \(u_a\) and \(u_b\)
were repeatedly paired.

\section{Proteinoid networks as physical reservoirs}

These adaptive dynamics can be interpreted computationally when stimulation
is treated as input, the evolving electrical state as a physical
transformation, and selected observables as output. This interpretation
follows the broader traditions of unconventional, chemical and physical
computation \citep{Adamatzky2017,Stepney2018,Nakajima2020}.
The proteinoid network implements a time-dependent transformation
\(\mathcal{F}_{\mathcal{G}} : \stimulus_{[0,T]} \longmapsto \response_{[0,T]}.\)
For a finite task,
\(\widehat{z} = \mathcal{R} \bigl( \state(T) \bigr),\)
where \(\mathcal{R}\) is a readout function.

For discrete tasks, a readout can convert continuous electrical trajectories into Boolean or categorical outputs.
For binary inputs \(u_1,u_2\in\{0,1\}\), a Boolean output is
\(y = \Theta \left( \bm{w}^{\mathsf{T}} \bm{\phi}(\state) - \vartheta \right),\)
where \(\bm{\phi}(\state)\) is a feature vector extracted from the electrical
response.

For a finite labelled data set, the classification error is
\(\epsilon_{\mathrm{task}} = \frac{1}{M} \sum_{m=1}^{M} \mathbf{1} \left[ \widehat{z}_m\neq z_m \right].\)
For temporal tasks, the same network can be interpreted as a physical reservoir whose transient state separates input histories before a comparatively simple readout is applied \citep{Jaeger2001,Maass2002,LukoseviciusJaeger2009,Tanaka2019,Nakajima2020}.
The network state satisfies
\(\state(t+\Delta t) = \mathcal{F} \bigl( \state(t), \stimulus(t) \bigr),\)
and the output is
\(\response(t) = \bm{W}_{\mathrm{out}} \bm{\phi} \bigl( \state(t) \bigr).\)

\begin{definition}[Separation property]
A proteinoid network satisfies the separation property if distinct input histories produce distinguishable internal states.
\end{definition}

\begin{definition}[Approximation property]
A proteinoid network satisfies the approximation property if a readout from its internal states can approximate a specified class of input--output functions to arbitrary prescribed accuracy.
\end{definition}

Computation under ideal conditions is insufficient for a material definition of intelligence. The same task must remain meaningful when the substrate is noisy, partially damaged or dynamically reorganised; robustness to component loss is a central issue in complex-network function \citep{Albert2000}. Let \(\mathcal{P}\) denote a perturbation, such as removal of nodes, disruption
of couplings or addition of noise. Define task performance before and after perturbation as \(P_0\) and \(P_{\mathcal{P}}\).
The robustness index is
\(\mathcal{R} = \frac{ P_{\mathcal{P}} }{ P_0 }.\)
If the network recovers with time, define
\(\mathcal{R}(t) = \frac{ P_{\mathcal{P}}(t) }{ P_0 }.\)
The recovery time is
\(\tau_{\mathrm{rec}} = \inf \left\{ t>0: \mathcal{R}(t)\geq 1-\epsilon \right\}.\)

\begin{definition}[Functional self-recovery]
A proteinoid network exhibits functional self-recovery if, following a finite perturbation, its task performance returns to within \(\epsilon\) of its pre-perturbation value without external reconstruction.
\end{definition}

\section{A quantitative functional for proteinoid intelligence}

The preceding quantities measure different aspects of behaviour and should
not individually be identified with intelligence. This multidimensional view
is compatible with information-processing approaches that distinguish
storage, transfer, transformation and task performance
\citep{Lizier2012,Dambre2012}. We therefore combine them into a task-relative
functional that rewards sensing, memory, adaptation, information transfer,
performance and robustness while penalising energetic and structural cost.
We define the intelligence functional

\begin{equation}
\mathcal{I}_{P} = \alpha_I \mathcal{S} + \alpha_M \mathcal{M} + \alpha_A \mathcal{A} + \alpha_T \mathcal{T} + \alpha_P \mathcal{P} + \alpha_R \mathcal{R} - \lambda_E \mathcal{E} - \lambda_C \mathcal{C}. \label{eq:intelligence-functional}
\end{equation}

Here,
$\mathcal{S} = \MI(U;Y)$, 
$\mathcal{M} = \frac{MC}{MC_{\max}}$,
$\mathcal{A} = \text{normalised adaptive capacity}$,\\
$\mathcal{T} = \text{normalised directed information transfer}$, 
$\mathcal{P} = 1-\epsilon_{\mathrm{task}}$,
$\mathcal{R} = \text{robustness or recovery score}$,
$\mathcal{E} = \text{energetic cost}$,
$\mathcal{C} = \text{structural or dynamical complexity cost}$.
All components \(\mathcal{S},\mathcal{M},\mathcal{A},\mathcal{T},
\mathcal{P},\mathcal{R},\mathcal{E},\mathcal{C}\) are represented by
dimensionless normalised quantities in the interval \([0,1]\). The weights
satisfy
\[
\alpha_I,\alpha_M,\alpha_A,\alpha_T,\alpha_P,\alpha_R,\lambda_E,\lambda_C
\geq 0.
\]
For comparisons in which the relative contributions of the positive terms
are to be interpreted as fractions of a fixed total reward, one may
additionally impose
\[
\alpha_I+\alpha_M+\alpha_A+\alpha_T+\alpha_P+\alpha_R=1.
\]
The latter convention is optional: in task-specific applications the weights
may instead be selected independently, provided that the same normalisation
and weighting conventions are retained across the systems being compared.

\begin{definition}[Proteinoid intelligence index]
The value of \(\mathcal{I}_{P}\) in \cref{eq:intelligence-functional} is the
proteinoid intelligence index for a specified task, environment, observation
interval and set of normalisation conventions.
\end{definition}

\begin{remark}
The index is task-relative. A network may exhibit high intelligence for one
class of stimuli and low intelligence for another.
\end{remark}

Once an intelligence functional has been defined, adaptive organisation can be expressed as a constrained optimisation problem. We do not assume that a proteinoid network performs explicit optimisation; rather, the variational formulation states the direction in which adaptive changes would have to move for task-relative performance to improve.
We hypothesise that adaptive proteinoid organisation tends towards states
that improve information-processing performance subject to physical
constraints. Energy--information constraints are fundamental to physical computation \citep{Landauer1961,Bennett1982}.

The optimisation problem is
\(\max_{\mathcal{G},\,\bm{\theta},\,\mathcal{R}} \quad \mathcal{I}_{P}\)
subject to
\(\dot{\state} = \bm{F} \bigl( \state, \stimulus, \weights, \bm{\theta} \bigr),;\quad 0 \leq g_{ij} \leq g_{\max},;\quad \mathcal{E} \leq E_{\max},;\quad \sum_{i,j}g_{ij} \leq G_{\max}.\)
An equivalent free-functional form is
\(\mathcal{J} = \MI(U;Y) + \beta_M\mathcal{M} + \beta_A\mathcal{A} + \beta_P\mathcal{P} - \lambda_E\mathcal{E} - \lambda_G\|\weights\|_1 - \lambda_D\mathcal{D}_{\mathrm{diss}}.\)

\begin{conjecture}[Information--energy trade-off]
For a fixed proteinoid composition and environmental regime, there exists an
intermediate level of coupling and activity that maximises information
processing per unit energetic cost.
\end{conjecture}

\begin{conjecture}[Adaptive performance improvement]
Under repeated structured stimulation, activity-dependent plasticity may
drive the effective network parameters towards regimes of increased
task-relative \(\mathcal{I}_{P}\), provided that the induced changes improve
one or more constituent measures without disproportionate energetic or
structural cost.
\end{conjecture}

The scalar functional is useful for comparison, but it should not obscure qualitative differences between simple reactivity and richer adaptive behaviour. We therefore distinguish a hierarchy of operational capacities.

We distinguish several mathematical levels of proteinoid intelligence.

\begin{description}

    \item[Level 0: Reactivity.]
    The output changes in response to stimulation:
    \(p(Y\mid U=u_a) \neq p(Y\mid U=u_b).\)

    \item[Level 1: Discrimination.]
    Distinct stimuli generate statistically distinguishable response
    distributions, with \(\MI(U;Y)>0\) and classification performance
    significantly exceeding an appropriate null or chance baseline.

    \item[Level 2: Memory.]
    Present responses depend on previous inputs:
    \(\MI \left( U_{t-k};Y_t \right) >0.\)

    \item[Level 3: Adaptation.]
    The input--output map changes as a consequence of experience:
    \(p(Y_t\mid U_t,\mathcal{H}_{t-1}) \neq p(Y_t\mid U_t).\)

    \item[Level 4: Integration.]
    Network-level performance exceeds that of isolated components:
    \(\mathcal{I}_{P}^{\mathrm{network}} > \frac{1}{N} \sum_{i=1}^{N} \mathcal{I}_{P}^{(i)}.\)

    \item[Level 5: Goal-directed computation.]
    The network exhibits a statistically significant experience-dependent
    improvement in performance on a specified task over repeated trials
    \(n\), without requiring the learning trajectory to be strictly
    monotonic.

\end{description}

These levels also make explicit a distinction that is central to the framework.
The model predicts that intelligence is not proportional to electrical activity alone. A network displaying large-amplitude oscillations may encode little information if all nodes behave identically. Conversely, a lower-amplitude network may possess greater intelligence if it supports diverse, reproducible and stimulus-dependent trajectories.
We therefore distinguish
\(\text{activity} \neq \text{information} \neq \text{intelligence}.\)
Activity is the existence of measurable dynamics. Information requires statistical dependence between environmental and internal variables. Intelligence additionally requires memory, adaptation, functional performance and robustness.

For conceptual illustration, the structure of
\cref{eq:intelligence-functional} can be represented schematically as

\[
\mathcal{I}_{P}
\sim
\underbrace{\mathcal{S}}_{\text{sensing}}
+
\underbrace{\mathcal{M}}_{\text{memory}}
+
\underbrace{\mathcal{A}}_{\text{adaptation}}
+
\underbrace{\mathcal{P}}_{\text{performance}}
+
\underbrace{\mathcal{R}}_{\text{robustness}}
-
\underbrace{\mathcal{K}}_{\text{cost}},
\]
where the proportionality-like symbol emphasises that this is a conceptual
summary rather than a second definition of the index.

\section{Theoretical results}

\begin{proposition}[Joint-response information bound]
Let \(Y_1,\ldots,Y_N\) be responses of individual microspheres and let
\(Y_{\mathrm{net}}=(Y_1,\ldots,Y_N)\) denote the complete joint response
vector, or more generally any representation from which every \(Y_i\) is
recoverable. Then
\[
\MI \left( U; Y_{\mathrm{net}} \right)
\geq
\max_i \MI(U;Y_i).
\]
\end{proposition}

\begin{proof}
Each individual response \(Y_i\) is a deterministic function of
\(Y_{\mathrm{net}}\). The Markov relation
\(U \rightarrow Y_{\mathrm{net}} \rightarrow Y_i\) therefore holds for every
\(i\), and the data-processing inequality gives
\(\MI(U;Y_i)\leq \MI(U;Y_{\mathrm{net}})\). Taking the maximum over \(i\)
yields the result.
\end{proof}

\begin{proposition}[Upper bound on memory]
For a stable linear \(N\)-dimensional network with a linear readout and
independent input, the total linear memory capacity is bounded by the
effective dimension of the observable state space.
\end{proposition}

\begin{proof}
For a stable linear system driven by an independent scalar input, every
linear reconstruction of a delayed input is an orthogonal projection onto
the observable state space. After whitening the state covariance, the sum
of the squared correlations associated with mutually independent delayed
inputs cannot exceed the rank of that space. The sum of the linear memory
coefficients is therefore bounded by the effective observable dimension,
which is at most \(N\).
\end{proof}

\begin{conjecture}[Trade-off between stability and sensitivity]
Increasing coupling strength initially improves signal transmission and
collective sensitivity, but sufficiently strong coupling reduces state-space
diversity through synchronisation. The location and even the existence of a
single optimum depend on the distribution of intrinsic time scales, network
topology and the form of coupling; the conjecture is therefore intended as a
testable statement for specified model classes rather than as a universal
monotonicity theorem.
\end{conjecture}

\begin{theorem}[Existence of an optimal admissible network]
Assume that the intelligence functional
\(\mathcal{I}_{P}(\weights)\) is continuous on the compact admissible set
\(\Omega = \left\{ \weights: 0\leq g_{ij}\leq g_{\max}, \quad \sum_{i,j}g_{ij}\leq G_{\max} \right\}.\)
Then there exists at least one admissible coupling matrix
\(\weights^{\ast}\in\Omega\) such that

\(\mathcal{I}_{P}(\weights^{\ast}) = \max_{\weights\in\Omega} \mathcal{I}_{P}(\weights).\)
\end{theorem}

\begin{proof}
The result follows from the extreme value theorem because
\(\mathcal{I}_{P}\) is continuous and \(\Omega\) is compact.
\end{proof}

The theorem establishes existence only; it does not imply uniqueness,
global dynamical accessibility, or that the optimum lies in the interior of
the admissible parameter space.

\section{Dimensionless analysis}

The preceding results involve parameters with different physical units and
time scales. Dimensionless groups make it possible to compare regimes across
compositions, geometries and experimental platforms without fixing a
particular absolute scale.
Introduce a characteristic voltage \(V_c\), time \(T_c\), coupling
\(g_c\) and stimulus \(U_c\). Define
\(\widetilde{V} = \frac{V}{V_c}, \qquad \widetilde{t} = \frac{t}{T_c}, \qquad \widetilde{g} = \frac{g}{g_c}, \qquad \widetilde{u} = \frac{u}{U_c}.\)
Relevant dimensionless groups include
\(\Pi_1 = \frac{\tau_a}{\tau_e}, \text{adaptation-to-excitation ratio},;\quad \Pi_2 = \frac{\tau_m}{\tau_e}, \text{memory-to-excitation ratio},;\quad \Pi_3 = \frac{g}{g_c}, \text{normalised coupling},;\quad \Pi_4 = \frac{\lambda}{d_{\mathrm{net}}}, \text{relative attenuation length},;\quad \Pi_5 = \frac{\sigma_{\mathrm{noise}}}{A_{\mathrm{signal}}}, \text{inverse signal-to-noise ratio},;\quad \Pi_6 = \frac{\tau_g}{\tau_e}, \text{plasticity-to-electrical timescale ratio}.\)
The dynamical regimes of the system may be represented in a phase diagram
over selected pairs of these parameters.

\section{Parameter estimation and model comparison}

Let measurements consist of time series
\(\mathcal{D} = \left\{ \bm{y}^{(r)}(t_k), \bm{u}^{(r)}(t_k) \right\},\)
where \(r\) indexes experimental replicates.
Parameters are estimated by minimising
\(\widehat{\bm{\theta}} = \arg\min_{\bm{\theta}} \sum_{r,k} \left\| \bm{y}^{(r)}(t_k) - \widehat{\bm{y}}^{(r)} \bigl( t_k;\bm{\theta} \bigr) \right\|^2 + \lambda_{\theta} \mathcal{R}_{\theta} \bigl( \bm{\theta} \bigr).\)
A probabilistic formulation is
\(p(\bm{\theta}\mid\mathcal{D}) \propto p(\mathcal{D}\mid\bm{\theta}) p(\bm{\theta}).\)
Model comparison may use
\(\mathrm{AIC} = 2k-2\log\widehat{\mathcal{L}},\)
or
\(\mathrm{BIC} = k\log n-2\log\widehat{\mathcal{L}}.\)

\section{Testable predictions}

The purpose of the mathematical framework is not merely descriptive. Its
usefulness depends on whether it produces predictions that can be contradicted
or supported experimentally. The principal predictions are stated below.

\begin{prediction}[Nonlinear dose response]
Proteinoid responses will exhibit a saturating nonlinear dependence on
stimulus intensity consistent with \cref{eq:hill-response}.
\end{prediction}

\begin{prediction}[Distance-dependent attenuation]
The amplitude of propagated electrical events will decrease approximately
exponentially with effective path length, as in \cref{eq:attenuation}.
\end{prediction}

\begin{prediction}[History dependence]
The response to an identical test stimulus will depend significantly on the
timing and identity of previous stimuli.
\end{prediction}

\begin{prediction}[Recovery law]
Following repeated stimulation, response recovery will exhibit an
approximately exponential dependence on the inter-stimulus interval, as in
\cref{eq:recovery-law}.
\end{prediction}

\begin{prediction}[Network advantage]
Coupled microsphere ensembles will exhibit greater stimulus classification
accuracy and memory capacity than isolated microspheres under matched
experimental conditions.
\end{prediction}

\begin{prediction}[Intermediate coupling optimum]
Information transmission and task performance will be maximal at an
intermediate coupling strength: weak coupling will prevent integration,
whereas excessive coupling will reduce dynamical diversity through
synchronisation.
\end{prediction}

\begin{prediction}[Activity-dependent reinforcement]
Repeated use of particular signal pathways will increase their effective
coupling and reduce transmission delay or attenuation.
\end{prediction}

\begin{prediction}[Critical regime]
Dynamic range, information transfer and response diversity are predicted to
exhibit enhanced values, potentially reaching a maximum, near a transition
between quiescent and globally synchronised activity.
\end{prediction}

\begin{prediction}[Energetic trade-off]
The highest information-per-energy ratio will not coincide with the largest
absolute response amplitude.
\end{prediction}

\begin{prediction}[Damage recovery]
Following partial disruption of the network, functional performance will
recover more rapidly in adaptive networks than in networks with fixed
coupling.
\end{prediction}

\section{Numerical experiments}

The same predictions can be examined numerically before or alongside
laboratory validation. Numerical experiments should be used to determine
which behaviours are generic consequences of the model and which require
narrow parameter choices.

The mathematical model should be evaluated through a sequence of numerical
experiments that follows the logic of the theory rather than treating the
model as a single monolithic simulation. The first stage is to reproduce
spontaneous activity in isolated microspheres and to determine whether the
chosen local dynamics can reconstruct experimentally observed spike trains
and oscillations. The second stage introduces coupling and follows localised
stimulation through one-dimensional, two-dimensional and random networks,
measuring propagation and attenuation while coupling strength is varied. This
naturally leads to an analysis of synchronisation and to estimation of
stimulus discrimination, mutual information, and linear and nonlinear memory
capacities. Once these baseline properties are established, the same model can
be driven repeatedly to examine habituation, sensitisation and recovery, and
the coupling matrix can be allowed to evolve in order to test whether
activity-dependent plasticity changes information flow or task performance.
Finally, node and edge removal provide controlled perturbations from which
robustness and self-recovery can be measured, while optimisation of the
intelligence functional tests whether the competing requirements of
integration, diversity, memory, performance and energetic cost produce an
intermediate operating regime rather than a trivial maximum at the strongest
possible activity.

Network geometry must be varied systematically because topology is itself a
candidate determinant of memory, integration and robustness. Linear chains
and rings provide simple geometries in which propagation paths are readily
interpreted; two-dimensional lattices and random geometric graphs provide
closer approximations to spatially embedded assemblies; Erd\H{o}s--R\'enyi,
small-world and scale-free networks \citep{ErdosRenyi1959,WattsStrogatz1998,BarabasiAlbert1999,Newman2003} permit the effects of heterogeneous
connectivity to be separated from purely geometric effects. Where microscopy
or other structural measurements permit it, these idealised graphs should be
compared with experimentally reconstructed proteinoid networks. The purpose of
this comparison is not to nominate a universally optimal topology, but to
identify which dynamical and informational properties are robust to network
architecture and which depend critically on a particular organisation of
couplings.

For numerical reproducibility, each reported regime should be evaluated over
multiple initial conditions and noise realisations. Parameter sweeps should
report the explored intervals and resolution, and stochastic results should
be summarised by central estimates together with dispersion or confidence
intervals. Training and evaluation data for reservoir readouts should be
separated, and null or surrogate inputs should be included when estimating
information-theoretic quantities. These requirements are part of the
mathematical validation of the framework rather than optional implementation
details.

The same observables should be extracted from each topology and parameter
regime. At the level of electrical activity these include spike rate and the
inter-spike interval distribution. Spatial transmission is characterised by
propagation velocity and attenuation length, while collective organisation is
captured by the synchronisation index. Information processing is quantified
through mutual information and transfer entropy, memory through the memory
capacity, and task performance through classification accuracy. Robustness
measures the persistence or recovery of function after perturbation, and the
energy-normalised intelligence measure places these benefits against their
physical cost. Reporting these quantities together is essential because a
network with the largest electrical response need not be the network with the
largest information content, and neither need be the network with the best
task-relative intelligence.

\section{Model scope and identifiability}

The framework is deliberately phenomenological. Different local dynamical
models may reproduce similar voltage traces, and information-theoretic
statistics estimated from finite noisy recordings do not by themselves
identify a unique physical mechanism. Consequently, model selection should
compare alternative local dynamics, coupling laws and memory models under
the same data partitions and complexity penalties. Parameters should be
reported with uncertainty, and conclusions should be restricted to features
that remain stable under plausible alternative model specifications. This
distinction between dynamical adequacy and mechanistic identification is
particularly important when the framework is fitted to proteinoid
experiments.

\section{Discussion}

The framework developed here treats proteinoid intelligence as an emergent property of excitable, adaptive and information-processing matter. Its starting point is deliberately conservative. Electrical activity by itself is not taken as evidence of intelligence, nor is a complicated waveform assumed to imply computation. A proteinoid system becomes interesting in the present sense only when its dynamics carry reproducible information about an external condition, when part of that information persists long enough to influence later behaviour, and when the resulting state can be used to produce a task-dependent response. Mutual information, memory capacity, stimulus discrimination, adaptation, task performance, robustness and energetic cost therefore describe different aspects of the same process rather than independent badges of intelligence \citep{Shannon1948,CoverThomas2006,Schreiber2000,Lizier2012,Dambre2012}.

This position places proteinoids within a broader movement away from substrate-specific definitions of computation and cognition. Embodied and morphological approaches emphasise that information processing can be distributed through the physical organisation of a system rather than localised in an abstract controller \citep{PfeiferBongard2007}. Material-computation and physical-reservoir approaches go further by exploiting the dynamics of matter itself as the computational transformation \citep{Stepney2018,Tanaka2019,Nakajima2020}. Work on basal cognition similarly asks which adaptive capacities can exist in cells and multicellular systems without requiring a brain-like architecture \citep{Lyon2015,Levin2019,FieldsLevin2020}. Proteinoids occupy an unusually austere position in this landscape because their organisation is generated from thermally polymerised amino acids rather than from evolved neural tissue. They therefore offer an experimental route for asking how far dynamical prerequisites of intelligent behaviour can be pushed towards chemically simple matter.

Several time scales are central to this question. Excitation and oscillation describe relatively fast electrical dynamics; adaptation and fading memory describe the persistence of recent history; plastic changes in coupling introduce a slower form of network memory; and growth or pruning can alter architecture on still longer time scales. Intelligence, if the term is useful here, is expected to arise from interactions among these scales rather than from any single one. This is consistent with recurrent and reservoir-computing theory, where useful computation depends on a compromise between nonlinear transformation, memory and stability \citep{Jaeger2001,Maass2002,LukoseviciusJaeger2009,Dambre2012}. Longer memory is not automatically advantageous, because excessive persistence can reduce sensitivity to changed inputs; likewise, stronger nonlinearity can improve separation while eventually destroying reproducibility.

The same compromise appears in collective synchronisation. Kuramoto's self-entrainment theory showed how increasing interaction strength can drive a population of oscillators from incoherence towards collective phase organisation \citep{Kuramoto1975,Kuramoto1984}. The Kuramoto order parameter is therefore useful here not because proteinoid microspheres are assumed to be ideal phase oscillators, but because it supplies a coarse measure of the degree to which local oscillatory dynamics become collectively coordinated. The extensive theory of the Kuramoto model confirms that the transition depends on frequency distributions, coupling architecture and heterogeneity \citep{Acebron2005,Arenas2008,Rodrigues2016}. Of particular relevance, Kuramoto and Battogtokh demonstrated that nonlocal coupling can sustain coherent and incoherent domains simultaneously \citep{KuramotoBattogtokh2002}. Such partially synchronised regimes are conceptually attractive for material intelligence: coherent regions can support integration and propagation, whereas incoherent or weakly coherent regions preserve state diversity. The experimentally relevant prediction is therefore not that maximal synchronisation maximises proteinoid intelligence, but that useful information processing may occur in an intermediate or spatially heterogeneous coordination regime.

This interpretation also connects the proposed critical-regime conjecture with broader work on excitable networks. Studies of neuronal avalanches and excitable-network models have suggested that regimes close to collective transitions can support large dynamic range and rich information transmission \citep{BeggsPlenz2003,KinouchiCopelli2006,ShewPlenz2013}. These results should not be transferred uncritically to proteinoids: the microscopic physics, geometry and noise sources are different. Nevertheless, they motivate a concrete experimental programme in which coupling, stimulus strength and network density are varied while synchronisation, response diversity, mutual information and task performance are measured simultaneously. A peak in one observable alone would be insufficient; the stronger result would be a reproducible region in which several information-processing measures improve without collapse into global uniformity.

Memory provides a second boundary between simple responsiveness and richer adaptive behaviour. The exponential memory variable used here is deliberately minimal and corresponds to a fading-memory process in which the influence of remote inputs decays with time \citep{BoydChua1985}. Reservoir-computing theory demonstrates why such fading memory can be computationally valuable: the present state can contain a compressed trace of recent input history while remaining responsive to new stimulation \citep{Jaeger2001,Maass2002}. Memory-capacity measures make this idea experimentally testable rather than rhetorical \citep{Jaeger2002memory,White2004,Dambre2012}. In proteinoid experiments, the important comparison is therefore between isolated microspheres and coupled ensembles under identical input histories. A collective memory advantage would require the network state to reconstruct past inputs more accurately or over longer delays than could be achieved by individual components.

Adaptation and plasticity introduce a stronger criterion. A response that merely relaxes after stimulation can be described by ordinary dynamical memory; a system that changes its future input--output relation as a consequence of repeated experience exhibits a form of adaptive memory. The correlation-dependent coupling laws introduced here are phenomenological analogues of activity-dependent plasticity rather than claims that proteinoids implement a biological Hebbian mechanism \citep{Hebb1949,GrossBlasius2008}. Their purpose is to generate falsifiable consequences: repeated use of a pathway should alter effective coupling, information flow, delay, attenuation or later task performance. If no persistent pathway-dependent change is observed, structural plasticity should be removed from the model rather than retained as a metaphor.

Information-theoretic quantities are essential for making these distinctions. Mutual information tests whether the response carries information about stimulation \citep{Shannon1948,CoverThomas2006}; transfer entropy asks whether one recorded component improves prediction of another beyond the target's own history \citep{Schreiber2000,Bossomaier2016}; active information storage measures predictive information retained locally \citep{Lizier2012}. None of these quantities alone establishes intelligence. Their role is to decompose a distributed dynamical process into experimentally measurable functions. In particular, high transfer entropy without task performance may indicate communication but not useful computation, whereas high classification accuracy obtained from a static stimulus artefact would not demonstrate memory or adaptation. The intelligence functional is intended to prevent such single-metric interpretations.

The reservoir-computing interpretation provides a practical benchmark for future experiments. If proteinoid networks are useful computational substrates, their transient electrical states should separate distinct input histories while retaining sufficient fading memory for a simple readout to reconstruct or classify those histories \citep{Jaeger2001,Maass2002,LukoseviciusJaeger2009,Tanaka2019,Nakajima2020}. Crucially, the readout should remain simpler than the substrate dynamics; otherwise the apparent computational performance may reside predominantly in post-processing. Comparisons should therefore include linear readouts, shuffled or surrogate controls, isolated microspheres, fixed versus adaptive networks and conventional nonlinear baselines. Demonstrating that a simple readout systematically extracts information unavailable from individual elements would provide much stronger evidence for collective material computation than visually complex voltage traces.

Energetic cost is equally important. Physical computation cannot be separated from dissipation and resource constraints \citep{Landauer1961,Bennett1982}. The framework therefore does not reward unlimited activity. A proteinoid network that doubles classification performance at the cost of a hundredfold increase in energy or material complexity may be less interesting than a quieter network that achieves moderate performance efficiently. This motivates reporting information-per-energy and performance-per-material measures alongside absolute task scores. It also provides a possible connection between primitive intelligence and physical economy: adaptive organisation may favour states that exploit the intrinsic dynamics of matter rather than imposing energetically expensive control.

Several parts of the framework remain phenomenological and require direct experimental validation. The precise state variables governing proteinoid excitability have not yet been identified, the physical mechanisms by which electrical activity is transmitted between microspheres are not fully established, and the relation between electrical dynamics, chemical transport and structural reorganisation remains uncertain. For these reasons the equations should be read as candidate mathematical laws and model classes, not as a claim that one microscopic mechanism has already been established. The FitzHugh--Nagumo, reaction--diffusion and Kuramoto descriptions are phenomenological reductions selected because they provide experimentally interpretable parameters and established analytical machinery \citep{FitzHugh1961,Nagumo1962,Murray2002,KeenerSneyd2009,Kuramoto1975,Kuramoto1984}. Their value lies in making assumptions explicit. Competing local models can be fitted to the same recordings, alternative coupling laws can be compared, and predictions about attenuation, recovery, memory and network advantage can be rejected if the data do not support them.

The framework also makes clear what would count as evidence for a genuinely collective effect. A network should not be judged superior merely because it contains more microspheres or produces a larger signal. The relevant question is whether coupling creates access to information, memory or task performance that is unavailable to isolated elements under matched conditions. This is why the joint readout, transfer measures, memory capacity and perturbation experiments are central to the proposed programme. Complex-network theory also predicts that topology can influence synchronisation, robustness and information pathways \citep{Newman2003,Boccaletti2006,Albert2000}. Systematic comparison of chains, lattices, random geometric networks, small-world networks and experimentally reconstructed assemblies should therefore distinguish effects of mere network size from effects of organisation \citep{ErdosRenyi1959,WattsStrogatz1998,BarabasiAlbert1999}.

The broader implication is that neurons need not be assumed in advance as the minimal substrate of adaptive information processing, a position increasingly explored in work on basal cognition, unconventional computation and material intelligence \citep{Lyon2015,Levin2019,FieldsLevin2020,Stepney2018,Nakajima2020}. If proteinoid microspheres can encode environmental variables, preserve traces of previous activity, modify their effective couplings and exploit those changes to alter future task performance, then a restricted but experimentally defensible form of intelligence can be discussed at the level of excitable chemical matter. Such a result would not erase the enormous difference between proteinoid networks and biological nervous systems. Rather, it would identify a set of dynamical primitives---excitability, memory, coupling, information integration, adaptation and constrained performance---from which more elaborate forms of intelligence may be constructed.

This interpretation also suggests an origins-of-intelligence question distinct from the classical origins-of-life question. Proteinoids were historically investigated because simple physicochemical processes can generate organised microspherical structures \citep{FoxHarada1958,FoxHarada1960,Fox1965,Fox1992}. If such structures additionally support nonlinear electrical dynamics and history-dependent collective responses, then self-assembly and primitive information processing need not be treated as completely separate stages. The experimentally important question becomes whether increasing organisation produces a measurable sequence from reactivity to discrimination, from discrimination to memory, and from memory to adaptive task performance. The hierarchy proposed in this paper is intended to make that sequence testable.

Ultimately, the significance of proteinoid intelligence will depend less on terminology than on comparative experiments. The strongest evidence would consist of reproducible stimulus encoding, measurable fading memory, experience-dependent modification, network-level information gain, robustness to perturbation and task performance that exceeds matched controls, all obtained under explicit energetic and structural constraints. Conversely, failure to demonstrate these properties would delimit the computational capabilities of the substrate and refine the theory. In either case, proteinoid networks provide a useful experimental system in which claims about primitive intelligence can be converted from analogy into quantitative hypotheses.

\section{Conclusions}

We have formulated information processing in proteinoid networks as an
applied problem in nonlinear dynamics, adaptive networks and information
theory. Proteinoid microspheres are represented as nonlinear excitable
elements embedded in adaptive weighted graphs, allowing local excitation,
signal propagation, attenuation, recovery, memory, structural plasticity,
information transfer and computation to be treated within a common
dynamical description.

Proteinoid intelligence is defined as a task-relative combination of sensing, memory, adaptation, information integration, functional performance and robustness under energetic and structural constraints. A quantitative
intelligence functional provides a basis for comparing individual microspheres, networks, compositions and experimental conditions.

The theory leads to explicit predictions concerning propagation,
history-dependent responses, network memory, optimal coupling, adaptive reinforcement, critical dynamics and information--energy trade-offs. Its value will ultimately be determined through confrontation with experimental data. If supported, the framework would contribute to a broader mathematical 
understanding of intelligence as a property of excitable matter rather than an exclusive product of nervous systems.

% Acknowledgements

% Data availability

\section*{Data Availability}

No new experimental data were generated or analysed in the present
theoretical study.

% Conflict of interest

\section*{Conflict of Interest}

The authors declare no conflict of interest.

% References

%\bibliographystyle{unsrtnat}
%\bibliography{proteinoid-intelligence-expanded}

\end{document}